\documentclass[aps,preprint,letterpaper,balancelastpage ,superscriptaddress,showpacs,prd]{revtex4}
\usepackage{amsfonts}
\usepackage{amsmath}
\usepackage{amssymb}
\usepackage{eurosym}
\usepackage{paralist}
\usepackage{graphicx}
\usepackage[colorlinks=true,linkcolor=blue,citecolor=blue, filecolor=blue,urlcolor=blue]{hyperref}

\usepackage{amsthm}

\newtheorem*{theorem}{Theorem}

\begin{document}

\title{Stellar equilibrium in Einstein-Chern-Simons gravity}

\author{C.A.C. Quinzacara}
\email{crcortes@udec.cl}
\affiliation{Departamento de F\'{\i}sica, Universidad de Concepci\'{o}n, Casilla 160-C,
Concepci\'{o}n, Chile}

\author{P. Salgado}
\email{pasalgad@udec.cl}
\affiliation{Departamento de F\'{\i}sica, Universidad de Concepci\'{o}n, Casilla 160-C,
Concepci\'{o}n, Chile}

\begin{abstract}
We consider a spherically symmetric internal solution within the context of Einstein-Chern-Simons gravity and derive a generalized five-dimensional Tolman-Oppenheimer-Volkoff (TOV) equation. It is shown that the generalized TOV equation leads, in a certain limit, to the standard five-dimensional TOV equation.
\end{abstract}


\maketitle

\section{\textbf{Introduction}}

Some time ago was shown that the standard, five-dimensional General Relativity can be obtained from Chern-Simons gravity theory for a certain Lie algebra $\mathfrak{B}$ \cite{IMPRS09}, which was obtained from the AdS algebra and a particular semigroup $S$ by means of the S-expansion procedure introduced in Refs. \cite{IRS06,IPRS09}. 

The five-dimensional Chern-Simons Lagrangian for the $\mathfrak{B}$ algebra is given by \cite{IMPRS09}

\begin{equation}
L_{\mathrm{ChS}}^{(5)}=\alpha _{1}l^{2}\varepsilon
_{abcde}R^{ab}R^{cd}e^{e}+\alpha _{3}\varepsilon _{abcde}\left( \frac{2}{3}%
R^{ab}e^{c}e^{d}e^{e}+2l^{2}k^{ab}R^{cd}T^{\text{ }e}+l^{2}R^{ab}R^{cd}h^{e}%
\right) ,  \label{1}
\end{equation}
where $l$ is a length scale in the theory (see \cite{IMPRS09}), $R^{ab}= d\omega^{ab}+\omega^a_{\phantom{a}c}\omega^{cb}$ is the curvature two-form with $\omega^{ab}$ the spin connection, and $T^a = De^a$ with $D$ the covariant derivative with respect to the Lorentz piece of the connection.

From (\ref{1}) we can see that \cite{IMPRS09}:
\begin{enumerate}[(i)]
\item the Lagrangian is split into two independent pieces, one proportional to $\alpha_1$ and the other to $\alpha_3$. If one identifies the field $e^a$ with the vielbein, the piece proportional to $\alpha_3$ contains the Einstein-Hilbert term $\varepsilon_{abcde}R^{ab}e^ce^de^e$ plus non-linear couplings between the curvature and the bosonic ``matter" fields $h^a$ and $k^{ab} =-k^{ba}$, which transform as a vector and as a tensor under local Lorentz transformations, respectively.

\item In the strict limit where the coupling constant $l$ equals zero we obtain solely the Einstein–Hilbert term in the Lagrangian \cite{IMPRS09}.

\item In the five-dimensional case, the connection of Eq. (17) of Ref. \cite{IMPRS09} has two possible candidates to be identified with the vielbein (see \cite{EHTZ06}), namely, the fields $e^a$ and $h^a$, since both transform as vectors under local Lorentz transformations. Choosing $e^a$, makes the
Einstein-Hilbert term to appear in the action, and $T^a = De^a$ can be interpreted as the torsion two-form. This choice brings in the Einstein equations.
\end{enumerate}

It is the purpose of this letter to find the stellar interior solution of the Einstein-Chern-Simons field equations, which were obtained in Refs. \cite{salg1,salg2}. 

We derive the generalized five-dimensional Tolman-Oppenheimer-Volkoff (TOV) equation and then we show that this generalized TOV equation leads, in a certain limit, to the standard five-dimensional TOV equation.

\section{Einstein-Chern-Simons field equations for a spherically symmetric metric}

In Ref. \cite{salg1} it was found that in the presence of matter described by the Lagrangian $L_\text{M} = L_\text{M}(e^a,h^a,\omega{ab})$, we see that the corresponding field equations are given by
\begin{align}
\varepsilon _{abcde}R^{cd}T^{e}& =\ 0,  \notag \\
\alpha _{3}l^{2}\varepsilon _{abcde}R^{bc}R^{de}& =-\frac{\delta L_\text{M}}{%
\delta h^{a}},  \notag \\
\varepsilon _{abcde}\left( 2\alpha _{3}R^{bc}e^{d}e^{e}+\alpha
_{1}l^{2}R^{bc}R^{de}+2\alpha _{3}l^{2}D_{\omega }k^{bc}R^{de}\right) & =-%
\frac{\delta L_\text{M}}{\delta e^{a}},  \label{2} \\
2\varepsilon _{abcde}\!\!\left( \alpha _{1}l^{2}R^{cd}T^{\text{{}}e}\!+\!\alpha
_{3}l^{2}D_{\omega }k^{cd}T^{e}\!+\!\alpha _{3}e^{c}e^{d}T^{e}\!+\!\alpha
_{3}l^{2}R^{cd}D_{\omega }h^{e}\!+\!\alpha _{3}l^{2}R^{cd}k_{\text{ }%
f}^{e}e^{f}\right) \!& =-\frac{\delta L_\text{M}}{\delta \omega ^{ab}}.  \notag
\end{align}

For simplicity we will assume $T^a = 0$ and $k^{ab} = 0$. In
this case the field equations (\ref{2}) can be written in the form \cite{salg2} 
\begin{align}
de^{a}+\omega _{\text{ }b}^{a}e^{b}& =0,  \notag \\
\varepsilon _{abcde}R^{cd}D_{\omega }h^{e}& =0,  \notag \\
\alpha _{3}l^{2}\star \left( \varepsilon _{abcde}R^{bc}R^{de}\right)& =-\star \left( \frac{\delta L_\text{M}}{\delta h^{a}}%
\right) ,  \label{3}\\
\star \left( \varepsilon _{abcde}R^{bc}e^{d}e^{e}\right)+\frac{1}{2\alpha}l^{2} \star \left( \varepsilon _{abcde}R^{bc}R^{de}\right)& =\kappa_\text{E} T_{ab}e^{b},  \notag 
\end{align}
where $\alpha =\alpha_3/\alpha_a$, $\kappa_\text{E} =\kappa/2\alpha_3$, $T_{ab} =−\star(\delta L_\text{M}/\delta e^a)$, ``$\star"$ is the Hodge star operator (see Appendix \ref{apen02}) and $T_{ab}$ is the energy-momentum tensor of matter fields (for details see Ref. \cite{salg2}). 

Since we are assuming spherical symmetry the metric will be of the form 
\begin{equation}
ds^{2}=-e^{2f(r)}dt^{2}+e^{2g(r)}dr^{2}+r^{2}d\Omega _{3}^{2}=\eta
_{ab}e^{a}e^{b}  \label{4}
\end{equation}
where $d\Omega _{3}^{2}=d\theta _{1}^{2}+\sin ^{2}\theta _{1}d\theta
_{2}^{2}+\sin ^{2}\theta _{1}\sin ^{2}\theta _{2}d\theta _{3}^{2}$ and $\eta
_{ab}=\mathrm{diag}(-1,+1,+1,+1,+1)$. The two unknown functions $f(r)$ and 
$g(r)$ will not turn out to be the same as in Ref. \cite{salg2}. In Ref. \cite{salg2} was found a spherically symmetric exterior solution, i.e. a solution where $\rho(r)= p(r)=0$. Now $f(r)$ and $g(r)$ must satisfy the field equations inside the star, where $\rho(r)\neq 0$ and $p(r)\neq0$. For this we need the energy-momentum tensor for the stellar material, which is taken to be a perfect fluid. 

Introducing an orthonormal basis,
\begin{equation*}
e^{T}=e^{f(r)}dt,\quad e^{R}=e^{g(r)}dr,\quad e^{1}=rd\theta _{1},\quad
e^{2}=r\sin \theta _{1}d\theta _{2},\quad e^{3}=r\sin \theta _{1}\sin \theta
_{2}d\theta _{3}.
\end{equation*}
Taking the exterior derivatives, using Cartan's first structural equation $T^a =de^a+\omega^a_{\phantom{b}b}e^b=0$ and the antisymmetry of the
connection forms we find the non-zero connection forms. The use of Cartan's second structural equation permits to calculate the curvature matrix $R^a_{\phantom{a} b}= d\omega^a_{\phantom{a} b}+\omega^a_{\phantom{a}c}\omega^c_{\phantom{b}}$.

Introducing these results in (\ref{3}) and considering the energy-momentum tensor as the energy-momentum tensor of a perfect fluid at rest, i.e., $T_{TT} = \rho(r)$ and $T_{RR} = T_{ii} = p(r)$, where $\rho(r)$ and $p(r)$ are the energy density and pressure (for the perfect fluid), we find \cite{salg2}
\begin{align}
\frac{e^{-2g}}{r^{2}}\left( g^{\prime }r+e^{2g}-1\right) +\textrm{sgn}(\alpha)l^2\frac{e^{-2g}}{r^{3}}g^{\prime }\left( 1-e^{-2g}\right)
& = 
\frac{\kappa_\text{E} }{12}\rho,  \label{36'} \\
\frac{e^{-2g}}{r^{2}}\left( f^{\prime }r-e^{2g}+1\right) +\textrm{sgn}(\alpha)l^2\frac{e^{-2g}}{r^{3}}f^{\prime }\left( 1-e^{-2g}\right) & =
\frac{\kappa_\text{E} }{12}p,  \label{37'} \\
\frac{e^{-2g}}{r^{2}}\left\{\left( -f^{\prime }g^{\prime
}r^{2}+f^{\prime \prime }r^{2}+\left( f^{\prime }\right)
^{2}r^{2}+2f^{\prime }r-2g^{\prime }r-e^{2g}+1\right)\right. \ \ \ \ \ \ \ \ \ \ \ \ \ \ &  \notag \\
\left.+\textrm{sgn}(\alpha)l^2 \left( f^{\prime \prime }+\left(
f^{\prime }\right) ^{2}-f^{\prime }g^{\prime }-e^{-2g}f^{\prime \prime
}-e^{-2g}\left( f^{\prime }\right) ^{2}+3e^{-2g}f^{\prime }g^{\prime
}\right)\right\}& =\frac{\kappa_\text{E} }{4}p.
\label{38'}
\end{align}

\section{The generalized Tolman-Oppenheimer-Volkoff equation}

Since when the torsion is null, the energy-momentum tensor satisfies the following condition (see Appendix \ref{apen02}):
\begin{equation}  \label{se06}
D_{\omega}(\star T_a)=0,
\end{equation}
we find that, for a spherically symmetric metric, (\ref{se06}) yields
\begin{equation}  \label{se07}
f^{\prime }(r)=-\frac{p^{\prime }(r)}{\rho(r)+p(r)},
\end{equation}
an expression known as the hydrostatic equilibrium equation. 

Following the usual procedure, we find that (\ref{36'}) has the following solution:
\begin{equation}\label{42'}
e^{-2g(r)}=1+\text{sgn}(\alpha)\frac{r^2}{l^2}-\text{sgn}(\alpha)\sqrt{\frac{%
r^4}{l^4}+\text{sgn}(\alpha)\frac{\kappa_\text{E}}{6\pi^2 l^2}\mathcal{M}(r)},
\end{equation}
where the Newtonian mass $\mathcal{M}(r)$ is given by
\begin{equation}
\mathcal{M}(r)=2\pi ^{2}\int_{0}^{r}\rho (\bar{r})\bar{r}^{3}d\bar{r}.
\label{41'}
\end{equation}

On the other hand, from (\ref{37'}) we find that
\begin{equation}
\frac{df(r)}{dr}=f^{\prime }(r)=\text{sgn}(\alpha )\frac{\kappa _\text{E}\
p(r)r^{3}+12r(1-e^{-2g(r)})}{12l^{2}e^{-2g(r)}\left( 1-e^{-2g(r)}+\text{sgn}%
(\alpha )\frac{r^{2}}{l^{2}}\right) } . \label{se10}
\end{equation}

Introducing (\ref{se10}) into (\ref{se07}) we find
\begin{equation}
\frac{dp(r)}{dr}=p^{\prime }(r)=-\text{sgn}(\alpha )\frac{\left( \rho (r)+p(r)\right) \left(
\kappa _\text{E}\ p(r)r^{3}+12r(1-e^{-2g(r)})\right) }{12l^{2}e^{-2g(r)}\left(
1-e^{-2g(r)}+\text{sgn}(\alpha )\frac{r^{2}}{l^{2}}\right) }  \label{se11}
\end{equation}
and introducing (\ref{42'}) into (\ref{se11}) we obtain the generalized five-dimensional Tolman-Oppenheimer-Volkoff equation
\begin{align}
\frac{dp(r)}{dr}=& -\frac{\kappa _\text{E}\ \mathcal{M}(r)\rho (r)}{12\pi ^{2}r^{3}}%
\left( 1+\frac{p(r)}{\rho (r)}\right) \left( 1+\text{sgn}(\alpha )\frac{%
\kappa _\text{E}}{6\pi ^{2}r^{4}}l^{2}\mathcal{M}(r)\right) ^{-1/2}  \notag
 \\
& \qquad \times \left[ \frac{\pi ^{2}r^{4}p(r)}{\mathcal{M}(r)}-\frac{12\ 
\text{sgn}(\alpha )\pi ^{2}r^{4}}{\kappa _\text{E}\ l^{2}\mathcal{M}(r)}\left( 1-%
\sqrt{1+\text{sgn}(\alpha )\frac{\kappa _\text{E}}{6\pi ^{2}r^{4}}l^{2}\mathcal{M}%
(r)}\right) \right]  \label{se13} \\
& \qquad \qquad \times \left[ 1+\text{sgn}(\alpha )\frac{r^{2}}{l^{2}}\left(
1-\sqrt{1+\text{sgn}(\alpha )\frac{\kappa _\text{E}}{6\pi ^{2}r^{4}}l^{2}\mathcal{%
M}(r)}\right) \right] ^{-1}.\notag
\end{align}

From (\ref{se13}) we can see that in the case of small $l^2$ limit,
we can expand the root to first order in $l^2$. In fact
\begin{equation}  \label{se14}
\sqrt{1+\text{sgn}(\alpha)\frac{\kappa_\text{E}}{6\pi^2r^4}l^2\mathcal{M}(r)}=1+\text{sgn}(\alpha)\frac{\kappa_\text{E}}{12\pi^2r^4}l^2\mathcal{M}(r)+\mathcal{O}(l^4).
\end{equation}

Introducing (\ref{se14}) into (\ref{se13}) we find
\begin{equation}\label{ectov01}
\frac{dp(r)}{dr}\approx-\frac{\frac{\kappa _\text{E}\ \mathcal{M}(r)\rho (r)}{12\pi ^{2}r^{3}}
\left( 1+\frac{p(r)}{\rho (r)}\right)\left(1+ \frac{\pi ^{2}r^{4}p(r)}{\mathcal{M}(r)}\right)}{\left(1+\text{sgn}(\alpha)\frac{\kappa_\text{E}}{6\pi^2r^4}l^2\mathcal{M}(r)\right)\left(1-\frac{\kappa_\text{E}}{12\pi^2r^2}\mathcal{M}(r)\right)}
\end{equation}

From (\ref{ectov01}) we can see that, in the limit where $l\longrightarrow0$, we obtain
\begin{equation}\label{ectov02}
\frac{dp(r)}{dr}=p^{\prime }(r)\approx -\frac{\kappa _{E}\mathcal{M}(r)}{12\pi ^{2}r^{3}}\left( 1+%
\frac{p(r)}{\rho (r)}\right) \left( 1+\frac{\pi ^{2}r^{4}p(r)}{\mathcal{M}(r)%
}\right) \left( 1-\frac{\kappa }{12\pi ^{2}r^{2}}\mathcal{M}(r)\right) ^{-1},
\end{equation}
which is the standard five-dimensional Tolman-Oppenheimer-Volkoff equation (see Eq. (\ref{A28})) (compare with the four-dimensional case shown in Ref. \cite{weinberg}).

To solve the generalized TOV equation (\ref{se13}), an equation of state relating $\rho$ and $p$ is needed. This equation should be supplemented by the boundary condition that $p(R) = 0$ where $R$ is the radius of the star.

Given an equation of state $p(\rho)$, the problem can be formulated as a pair of first-order differential equations for $p(r)$, $\mathcal M(r)$ and $\rho(r)$, (\ref{se13}) and
\begin{equation}
\mathcal{M}^{\prime}(r)=2\pi^{2}r^3\rho(r),  \label{se15}
\end{equation}
with the initial condition $\mathcal M(0) = 0$. In addition, it is necessary to provide the initial condition $\rho(0) = \rho_0$.

Let us return to the problem of calculating the metric. Once we compute $\rho(r)$, $\mathcal M(r)$, and $p(r)$, we can immediately obtain $g(r)$ from (\ref{42'}) and $f(r)$ from (\ref{se10})
\begin{align}
f(r)=& -\int_{r}^{\infty }\ \frac{\kappa _\text{E}\ \mathcal{M}(\bar{r})}{12\pi
^{2}\bar{r}^{3}}\left( 1+\text{sgn}(\alpha )\frac{\kappa _\text{E}}{6\pi ^{2}\bar{%
r}^{4}}l^{2}\mathcal{M}(\bar{r})\right) ^{-1/2}   \notag \\
& \qquad \times \left[ \frac{\pi ^{2}\bar{r}^{4}p(\bar{r})}{\mathcal{M}(\bar{%
r})}-\frac{12\ \text{sgn}(\alpha )\pi ^{2}\bar{r}^{4}}{\kappa _\text{E}\ l^{2}%
\mathcal{M}(\bar{r})}\left( 1-\sqrt{1+\text{sgn}(\alpha )\frac{\kappa _\text{E}}{%
6\pi ^{2}\bar{r}^{4}}l^{2}\mathcal{M}(\bar{r})}\right) \right] \label{se16}  \\
& \qquad \qquad \times \left[ 1+\text{sgn}(\alpha )\frac{\bar{r}^{2}}{l^{2}}%
\left( 1-\sqrt{1+\text{sgn}(\alpha )\frac{\kappa _\text{E}}{6\pi ^{2}\bar{r}^{4}}%
l^{2}\mathcal{M}(\bar{r})}\right) \right] ^{-1}\ d\bar{r}\notag 
\end{align}
where we have set $f(\infty )=0$, a condition consistent with the asymptotic limit from the exterior solution.

It should be noted that if $r > R$, i.e., out of the star, the following conditions are satisfied:
\begin{equation}  \label{se17}
\mathcal{M}(r)=M\quad,\quad p(r)=\rho(r)=0.
\end{equation}
Integrating (\ref{se16}) we find
\begin{equation}
f(r)=\frac{1}{2}\ln \left[ 1+\text{sgn}(\alpha )\frac{r^{2}}{l^{2}}\left( 1-%
\sqrt{1+\text{sgn}(\alpha )\frac{\kappa _\text{E}}{6\pi ^{2}r^{4}}l^{2}M}\right) %
\right] , \label{se19}
\end{equation}%
so that
\begin{equation}
e^{2f(r)}=e^{-2g(r)}=1+\text{sgn}(\alpha )\frac{r^{2}}{l^{2}}-\text{sgn}%
(\alpha )\sqrt{\frac{r^{4}}{l^{4}}+\text{sgn}(\alpha )\frac{\kappa _\text{E}}{%
6\pi ^{2}l^{2}}M}, \label{se20}
\end{equation}
which coincides with the outer solution.

\subsection{Constant Density: $\protect\rho(r)=\protect\rho_0$}

We will now consider the solution of (\ref{se13}) in the case where
the energy density is constant, $\rho(r) = \rho_0$, inside the star. In this case the hydrostatic equilibrium equation (\ref{se07}) can be directly integrated,
\begin{equation}
\rho _{0}+p(r)=Ce^{-f(r)}, \label{se21}
\end{equation}
where $C$ is an integration constant. 

On the other hand, from (\ref{se15}) $\mathcal{M}(r)$ is given by
\begin{equation}
\mathcal{M}(r)=\frac{\pi ^{2}}{2}\rho
_{0}r^{4} . \label{se22}
\end{equation}

Introducing (\ref{se22}) into (\ref{42'}) we have
\begin{equation}
e^{-2g(r)}= 1+\text{sgn}(\alpha )\frac{r^{2}}{l^{2}}-\text{sgn}(\alpha )%
\sqrt{\frac{r^{4}}{l^{4}}+\text{sgn}(\alpha )\frac{\kappa _\text{E}}{12l^{2}}\rho
_{0}\ r^{4}}.\label{se23} 
\end{equation}

Now, let us add the field equations (\ref{36'}) and (\ref{37'})
\begin{equation}
\frac{e^{-2g}}{r^{3}}(f^{\prime }+g^{\prime })\left[ r^{2}+\text{sgn}(\alpha
)l^{2}\left( 1-e^{-2g}\right) \right] =\frac{\kappa _\text{E}}{12}(\rho _{0}+p).
\label{se24}
\end{equation}%

Using now (\ref{se21}), multiplying by $e^{-g}$ and integrating we have
\begin{equation}
e^{f}=\frac{\kappa _\text{E}}{12}Ce^{-g}\int \frac{r^{3}\ dr}{e^{-3g}\left[ r^{2}+%
\text{sgn}(\alpha )l^{2}\left( 1-e^{-2g}\right) \right] }+C_{0}e^{-g},
\label{se27}
\end{equation}
where $C_0$ is the corresponding integration constant. Since 
\begin{equation}
\int \frac{r^{3}\ dr}{e^{-3g}\left[ r^{2}+%
\text{sgn}(\alpha )l^{2}\left( 1-e^{-2g}\right) \right] }=\frac{-\text{sgn}(\alpha )l^{2}e^{g(r)}}{\sqrt{1+\text{sgn}%
(\alpha )\frac{\kappa _\text{E}}{12}l^{2}\rho _{0}}\,\left( 1-\sqrt{1+\text{sgn}(\alpha )\frac{\kappa _\text{E}}{12}l^{2}\rho _{0}\ }\right) }
\end{equation}
we find
\begin{equation}
e^{f}=C_{1}+C_{0}e^{-g} , \label{se30}
\end{equation}%
where
\begin{equation}
C_{1}:=-\frac{\text{sgn}(\alpha )\kappa _\text{E}\, l^{2}\,C}{12\sqrt{1+\text{sgn}%
(\alpha )\frac{\kappa _{E}}{12}l^{2}\rho _{0}}\,\left( 1-\sqrt{1+\text{sgn}%
(\alpha )\frac{\kappa _{E}}{12}l^{2}\rho _{0}\ }\right) }.  \label{se31}
\end{equation}

Then, we proceed to adjust the constants $C$, $C_0$, and $C_1$, so that the interior solution and exterior must match at $r =R$. In addition one should require that the pressure vanishes at $r =R$. 

The calculations give
\begin{equation}  \label{se34}
C=\rho_0\sqrt{1+\text{sgn}(\alpha)\frac{R^2}{l^2}\left(1-\sqrt{1+\text{sgn}%
(\alpha)\frac{\kappa_\text{E}}{12}l^2\rho_0}\right)},
\end{equation}

\begin{equation}
C_{1}=-\frac{\text{sgn}(\alpha )\kappa _\text{E}\ l^{2}\rho _{0}\sqrt{1+\text{sgn}%
(\alpha )\frac{R^{2}}{l^{2}}\left( 1-\sqrt{1+\text{sgn}(\alpha )\frac{\kappa
_\text{E}}{12}l^{2}\rho _{0}}\right) }}{12\sqrt{1+\text{sgn}(\alpha )\frac{\kappa
_\text{E}}{12}l^{2}\rho _{0}}\left( 1-\sqrt{1+\text{sgn}(\alpha )\frac{\kappa _\text{E}%
}{12}l^{2}\rho _{0}\ }\right) }  ,\label{se35}
\end{equation}%
and
\begin{equation}
C_{0}=-\frac{1}{\sqrt{1+\text{sgn}(\alpha )\frac{\kappa _\text{E}}{12}l^{2}\rho
_{0}}} . \label{se36}
\end{equation}

\section{Summary and outlook}

We have considered a spherically symmetric internal solution within the context of Einstein-Chern-Simons gravity. We derived the generalized five-dimensional Tolman-Oppenheimer-Volkoff (TOV) equation and then we proved that this generalized TOV equation leads, in a certain limit, in the standard five-dimensional TOV equation.

\begin{acknowledgments}
This work was supported in part by Direcci\'on de Investigaci\'on, Universidad de Concepci\'on through Grant \# 212.011.056-1.0 and in part by FONDECYT through Grant N° 1130653. One of the authors (C.A.C.Q) was supported by grants from the Comisión Nacional de Investigación Cient\'ifica y Tecnol\'ogica CONICYT and from the Universidad de Concepci\'on, Chile.
\end{acknowledgments}

\appendix

\section{The standard Tolman-Oppenheimer-Volkoff equation in 5D}

Let us recall that the energy-momentum tensor satisfies the condition
\begin{equation}
\nabla_{\mu}T^{\mu\nu}=0.  \label{A25}
\end{equation}

If $T_{TT}=\rho(r)$ and $T_{RR} =T_{ii} =p(r)$ we find
\begin{equation*}
\nabla_{\mu}T^{\mu r}=\frac{f'(r)\Bigl(\rho(r)+p(r)\Bigr)+p'(r)}{e^{2g(r)}},
\end{equation*}
so that
\begin{equation}
f^{\prime }=-\frac{p^{\prime }}{\rho +p},  \label{A26}
\end{equation}%
an expression known as the \emph{hydrostatic equilibrium equation}.

From Eqs. (A10) and (A22) of Ref. \cite{salg2} we find
\begin{equation}
f^{\prime }(r)=\frac{\kappa _\text{E}\mathcal{M}(r)}{12\pi ^{2}r^{3}}\left( 1+%
\frac{\pi ^{2}r^{4}p(r)}{\mathcal{M}(r)}\right) \left( 1-\frac{\kappa_\text{E} }{
12\pi ^{2}r^{2}}\mathcal{M}(r)\right) ^{-1}.  \label{A27}
\end{equation}
Introducing (\ref{A26}) into (\ref{A27}) we obtain the standard five-dimensional \emph{Tolman-Oppenheimer-Volkoff} equation
\begin{equation}
p^{\prime }(r)=-\frac{\kappa _\text{E}\mathcal{M}(r)}{12\pi ^{2}r^{3}}\left( 1+%
\frac{p(r)}{\rho (r)}\right) \left( 1+\frac{\pi ^{2}r^{4}p(r)}{\mathcal{M}(r)%
}\right) \left( 1-\frac{\kappa_\text{E} }{12\pi ^{2}r^{2}}\mathcal{M}(r)\right) ^{-1}.
\label{A28}
\end{equation}%
This may be compared with the four-dimensional case shown in equation (1.11.13) of reference \cite{weinberg}.

\section{Energy-momentum tensor\label{apen02}}

It is known that if the torsion is null, then the energy-momentum tensor is divergence-free, $\nabla_{\mu}T^{\mu\nu}=0$. The 1-\emph{form} energy-momentum is given by
\begin{equation}  \label{tenem01}
\hat T_a:=T_{\mu\nu} e_a^{\mu}\ dx^{\nu}.
\end{equation}

\begin{theorem}
If the energy-momentum tensor $T_{\mu \nu }$ and the 1-form energy-momentum $\hat{T}_{a}$ are related by equation (\ref{tenem01}), then in a torsion-free space-time
\begin{equation}
\nabla _{\mu }T_{\text{ }\nu }^{\mu }=-e_{\nu }^{a}\star D_{\omega }(\star 
\hat{T}_{a})
\end{equation}
\end{theorem}

\begin{proof}
\begin{equation}
\star \hat{T}_{a}= \frac{\sqrt{-g}}{4!}\epsilon _{\mu \nu \rho \sigma \tau }T_{a}^{\text{ }%
\mu }\ dx^{\nu }dx^{\rho }dx^{\sigma }dx^{\tau }.
\end{equation}
After some algebra, we find
\begin{equation}
-e_{\nu }^{a}\star D_{\omega }(\star \hat{T}_{a})=\frac{1}{\sqrt{-g}}\partial _{\lambda }(\sqrt{-g})T_{\nu }^{\text{ \ }\lambda }+\partial _{\lambda }T_{\nu }^{\text{ \ }\lambda }-T_{a}^{\text{ \ }%
\lambda }\left( \partial _{\lambda }e_{\nu }^{a}+\omega _{\text{ }\lambda
b}^{a}e_{\nu }^{b}\right),
\end{equation}
and using the Weyl's lemma
\begin{equation}
\partial _{\lambda }e_{\nu }^{a}+\omega _{\text{ \ }\lambda b}^{a}e_{\nu
}^{b}-\Gamma _{\lambda \nu }^{\text{ \ }\rho }e_{\rho }^{a}=0,
\end{equation}
we obtain
\begin{equation}
-e_{\nu }^{a}\star D_{\omega }(\star \hat{T}_{a})= \frac{1}{\sqrt{-g}}\partial _{\lambda }(\sqrt{-g})T_{\nu }^{\text{ \ }\rho }+\partial _{\lambda }T_{\nu }^{\text{ \ }\lambda }-\Gamma _{\lambda \nu }^{\text{ \ \ }\rho} T_{\rho }^{\text{ \ }\lambda },
\end{equation}
\begin{equation}
-e_{\nu }^{a}\star D_{\omega }(\star \hat{T}_{a})= \partial _{\lambda
}T_{\nu }^{\text{ \ }\lambda }+\Gamma _{\lambda \rho }^{\text{ \ \ }\lambda
}T_{\nu }^{\text{ \ }\rho }-\Gamma _{\lambda \nu }^{\text{ \ }\rho }T_{\rho
}^{\text{ \ }\lambda } = \nabla _{\lambda }T_{\nu }^{\text{ \ }\lambda }.
\end{equation}

\end{proof}

\subsection{The Hodge star operator}

The Hodge star operator for a $p-$form $P=\frac{1}{p!}P_{\alpha_1\cdots \alpha_p}dx^{\alpha_1}\cdots dx^{\alpha_p}$ in a $d$-dimensional manifold with a non-singular metric tensor $g_{\mu\nu}$ is defined as
\begin{equation*}
\star P=\frac{\sqrt{|g|}}{(d-p)!p!}\varepsilon_{\alpha_1\cdots \alpha_d}g^{\alpha_1\beta_1}\cdots g^{\alpha_p\beta_p} P_{\beta_1\cdots \beta_p}dx^{\alpha_{p+1}}\cdots dx^{\alpha_d}\quad,
\end{equation*}
where $\varepsilon_{\alpha_1\cdots \alpha_d}$ is the total antisymmetric Levi-Civita tensor density of weight $-1$.

\subsection{Hydrostatic equilibrium equation}

Let us consider a spherically and static-symmetric metric in five dimensions. The $1$-form energy-momentum is given by
\begin{equation}
\hat T_a=T_{ab}e^b
\end{equation}
where $T_{ab}$ is the energy-momentum tensor in a comoving orthonormal frame. So, if the matter is a perfect fluid then  
\begin{equation}
T_{TT}=\rho(r)\quad,\quad T_{RR}=T_{ii}=p(r).
\end{equation}

Computing the conservation equation
\begin{equation}
D_{\omega }(\star \hat{T}_{a})=0
\end{equation}
we have
\begin{equation}
D_{\omega }(\star \hat{T}_{a})=D_{\omega }(T_{ab}\star e^{b})=\frac{1}{4!}%
\epsilon _{fbcde}(D_{\omega }T_{a}^{\text{ \ }f})e^{b}e^{c}e^{d}e^{e},
\end{equation}
where we have used the torsion-free condition $D_{\omega }e^{a}=0$. Therefore
\begin{equation}
D_{\omega }(\star \hat{T}_{a})=\frac{1}{4!}\epsilon _{fbcde}(dT_{a}^{\text{
\ }f}+\omega _{a}^{\text{ \ }g}T_{g}^{\text{ \ }f}+\omega _{\text{ }%
g}^{f}T_{a}^{\text{ \ }g})e^{b}e^{c}e^{d}e^{e}.
\end{equation}

The calculations give 
\begin{equation}
D_{\omega }(\star \hat{T}_{R})=e^{-g}\left( p^{\prime }+f^{\prime }(\rho
+p)\right) e^{T}e^{R}e^{1}e^{2}e^{3}=0
\end{equation}
from which we get the so-called hydrostatic equilibrium equation
\begin{equation}
p^{\prime }+f^{\prime }(\rho +p)=0.  \label{apeneec}
\end{equation}

\section{Dynamic of the field $h^a$}

We consider now the field $h^a$. Expanding the field $h^a = h^a_\mu\, dx^\mu$ in their holonomic index we have \cite{salg2}
\begin{equation}\label{campoha}
h_a=h_{\mu\nu}e_a^\mu\, dx^\nu
\end{equation}

For the space-time to be static and spherically symmetric, the field $h_{\mu\nu}$ must satisfy the Killing equation $\mathcal L_\xi h_{\mu\nu} =0$ for $\xi_0=\partial_t$ and the six generators of the sphere $S_3$ must be
\begin{align}
\xi_0&=\partial_t,\notag\\
\xi_1&=\partial_{\theta_3},\notag\\
\xi_2&=\sin\theta_3\ \partial_{\theta_2}+\cot\theta_2\cos\theta_3\ \partial_{\theta_3},\notag\\
\xi_3&=\sin\theta_2\sin\theta_3\ \partial_{\theta_1}+\cot\theta_1\cos\theta_2\sin\theta_3\ \partial_{\theta_2}+\cot\theta_1\csc\theta_2\cos\theta_3\ \partial_{\theta_3}\label{vk}\\
\xi_4&=\cos\theta_3\ \partial_{\theta_2}-\cot\theta_2\sin\theta_3\ \partial_{\theta_3},\notag\\
\xi_5&=\sin\theta_2\cos\theta_3\ \partial_{\theta_1}+\cot\theta_1\cos\theta_2\cos\theta_3\ \partial_{\theta_2}-\cot\theta_1\csc\theta_2\sin\theta_3\ \partial_{\theta_3},\notag\\
\xi_6&=\cos\theta_2\ \partial_{\theta_1}-\cot\theta_1\sin\theta_2\ \partial_{\theta_2}.\notag
\end{align}

Then, we have
\begin{align}
h^T&=h_{tt}(r)\ e^T+h_{tr}(r)\ e^R,\notag\\
h^R&=h_{rt}(r)\ e^T+h_{rr}(r)\ e^R, \label{caha01}\\
h^i&= h(r)\ e^i.\notag
\end{align}

From Eq. (\ref{3}) we know that the dynamic of the field $h^a$ is
given by
\begin{equation}\label{ecch01}
\epsilon_{abcde}R^{cd}D h^e=0
\end{equation}
with
\begin{equation}
D h^a=d h^a+\omega^a_{\phantom 1 b}h^b
\end{equation}
where
\begin{align}
D h^T&= e^{-g}\left(-h'_{tt}-f'h_{tt}+f'h_{rr}\right)\ e^Te^R,\label{C06}\\
D h^R&=e^{-g}\left(-h'_{rt}-f'h_{rt}+f'h_{tr}\right)\ e^Te^R,\label{C07}\\
D h^i&=\frac{e^{-g}}{r}\left(rh'+h-h_{rr}\right)\ e^Re^i-\frac{e^{-g}}{r}h_{rt}\ e^Te^i.\label{C08}
\end{align}

Introducing (\ref{C06} - \ref{C08}) into (\ref{ecch01}) we have
\begin{align}
h_{tr}&=h_{rt}=0,\label{C09}\\
h_r&=(rh)',\label{C10}\\
h'_t&=f'(h_r-h_t).\label{C11}
\end{align}

To find solutions to (\ref{C09}, \ref{C10}, \ref{C11}), we assume that
$h_t(r)$ depends on $r$ only through $f(r)$, namely
\begin{equation}\label{C12}
h_t(r)=h_t\Bigl(f(r)\Bigr)
\end{equation}
Introducing (\ref{C12}) into (\ref{C11}) we have
\begin{equation}
\frac{dh_t(f)}{df}f'(r)=f'(h_r-h_t)
\end{equation}
from which we obtain the following linear differential equation, which is of first order and inhomogeneous:
\begin{equation}\label{edohomo}
\dot h_t+h_t=h_r,
\end{equation}
where $\dot h_t:=\frac{dh_t(f)}{df}$. The homogeneous solution is given by
\begin{equation}
h_{t}^\text{h}(f)=Ae^{-f(r)},
\end{equation}
where $A$ is a constant to be determined. 

The particular solution depends on the shape of $h_r$. If we assume a functional relationship $h$ with $f$ , then the linearity of differential equation suggests the following ansatz:
\begin{equation}
h_r(r)=h_r\Bigl(f(r)\Bigr)=\sum_{n=0}^\infty B_ne^{n f(r)}+\sum_{m=2}^\infty C_{m}e^{-m f(r)},
\end{equation}
where $B_n$ and $C_m$ are real constants. So that the particular solution is given by
\begin{equation}
h_t^\text{p}(f)=\sum_{n=0}^\infty \frac{B_n}{n+1}e^{n f(r)}-\sum_{m=2}^\infty \frac{C_m}{m-1}e^{-m f(r)}.
\end{equation}

Therefore the general solution is of the form 
\begin{equation}
h_{t}\Bigl(f(r)\Bigr)=Ae^{-f(r)}+\sum_{n=0}^\infty \frac{B_n}{n+1}e^{n f(r)}-\sum_{m=2}^\infty \frac{C_m}{m-1}e^{-m f(r)}.
\end{equation}

From (\ref{C10}) we find
\begin{equation}\label{C19}
h(r)=\frac{1}{r}\left(\int h_r(r)\, dr+D\right),
\end{equation}
where $D$ is an integration constant. This means
\begin{equation}\label{C20}
h(r)=\frac{1}{r}\sum_{n=0}^\infty\left( B_n\int e^{n f(r)}\,dr\right)+\frac{1}{r}\sum_{m=2}^\infty \left(C_{m}\int e^{-m f(r)}\,dr\right)+\frac{D}{r},
\end{equation}
where $A$, $B_n$, and $C_m$ are arbitrary constants, and $-e^{2f(r)}$ is the metric coefficient $g_{00}$.

\subsection{Field asymptotically constant}

Consider the simplest case where
\begin{equation}
h_r(r)=h=\text{constant}
\end{equation}
in this case (\ref{C19}) leads
\begin{equation}
h(r)= h+\frac{D}{r}
\end{equation}
and
\begin{equation}
h_t(r)=Ae^{-f(r)}+h.
\end{equation}

Since the vielbein is regular at $r =0$ (center of the star), $h^a$ should also be regularly at $r = 0$, i.e. we should have $D =0$. Note that the coefficient $e^{f(r)}$ is regular at $r =0$ as can be seen from (\ref{C19}).

From (\ref{C20}) we can see that the asymptotic behavior of the metric coefficients is given by
\begin{equation}
e^{2f(r\rightarrow\infty)}=e^{-2g(r\rightarrow\infty)}=1.
\end{equation}

Thus the asymptotic behavior of the field $h^a$ is given by
\begin{equation}\label{C25}
h_r(r\rightarrow\infty)=h,\quad h(r\rightarrow\infty)= h,\quad h_t(r\rightarrow\infty)=A+h.
\end{equation}

\subsection{Constant density}

If the density is constant then the inner solution is given by (\ref{se23}) and (\ref{se30}). In this case the solution for the field $h^a$ is given by
\begin{equation}
h_r(r)=h,\quad h(r)= h
\end{equation}
and
\begin{equation}
h_t(r)=\left\{
\begin{array}{ll}
\displaystyle\frac{A}{C_0+C_1e^{-g(r)}}+h&\ \text{if } r<R,\\
\displaystyle\frac{A}{e^{-g(r)}}+h&\ \text{if } r\geq R,
\end{array} \right.
\end{equation}
where
\begin{equation}
e^{-g(r)}=\left\{
\begin{array}{ll}
\displaystyle\sqrt{1+\text{sgn}(\alpha )\frac{r^{2}}{l^{2}}-\text{sgn}%
(\alpha )\sqrt{\frac{r^{4}}{l^{4}}+\text{sgn}(\alpha )\frac{\kappa _\text{E}}{%
6\pi ^{2}l^{2}}\mathcal{M}(r)}} &\ \text{if } r<R\\\\
\displaystyle\sqrt{1+\text{sgn}(\alpha )\frac{r^{2}}{l^{2}}-\text{sgn}%
(\alpha )\sqrt{\frac{r^{4}}{l^{4}}+\text{sgn}(\alpha )\frac{\kappa _\text{E}}{%
6\pi ^{2}l^{2}}M}} &\ \text{if } r\geq R
\end{array} \right.
\end{equation}

\end{document}